\documentclass[final,3p,times,11pt]{elsarticle}
\usepackage[utf8]{inputenc}

\usepackage{amsmath,amssymb}
\usepackage{amsfonts}
\usepackage{amsthm}

\usepackage{bm}
\usepackage{graphicx}

\usepackage{enumitem}
\setlist{nolistsep}

\makeatletter
\def\ps@pprintTitle{%
\let\@oddhead\@empty
\let\@evenhead\@empty
\def\@oddfoot{\centerline{\thepage}}%
\let\@evenfoot\@oddfoot}
\makeatother

\newcommand{\beq}{\begin{equation}}
\newcommand{\eeq}{\end{equation}}

\DeclareMathOperator{\sech}{sech}
\newcommand{\p}{\partial}
\newcommand{\emn}{\epsilon_{\mu\nu}}

\newcommand{\ex}{\bm{\hat{e}}_1}
\newcommand{\ey}{\bm{\hat{e}}_2}
\newcommand{\ez}{\bm{\hat{e}}_3}

\newcommand{\Eex}{E_{\rm ex}}
\newcommand{\Ea}{E_{\rm an}}
\newcommand{\Edm}{E_{\rm DM}}
\newcommand{\magn}{\bm{m}}

\newcommand{\DM}{D}

\newcommand{\Anisotropy}{K}

\newcommand{\dmscaled}{\epsilon}

\newcommand{\ldw}{\ell_{\rm w}}
\newcommand{\ldk}{\ell_{S}}
\newcommand{\T}{\Theta}
\renewcommand{\t}{\theta}

\newcommand{\hf}{\textstyle{\frac{1}{2}}}

\newcommand{\der}{\mathrm{d}}

\newcommand{\dr}{\mathrm{d}r}
\newcommand{\dT}{\mathrm{d}T}
\newcommand{\dtau}{\mathrm{d}\tau}

\newtheorem*{theorem*}{Theorem}
\newtheorem{theorem}{Theorem}

\begin{document}

\begin{frontmatter}

\title{Chiral skyrmions of large radius}
\author{Stavros Komineas}
\address{Department of Mathematics and Applied Mathematics, University of Crete, 70013 Heraklion, Crete, Greece}
\author{Christof Melcher}
\address{Department of Mathematics \& JARA Fundamentals of Future Information Technology, RWTH Aachen University, 52056 Aachen, Germany}
\author{Stephanos Venakides}
\address{Department of Mathematics, Duke University, Durham, NC, USA}
\date{\today}

\begin{abstract}
We study the structure of an axially symmetric magnetic skyrmion in a ferromagnet with the Dzyaloshinskii-Moriya interaction.
We examine the regime of large skyrmions and we identify rigorously the critical value of the dimensionless parameter at which the skyrmion radius diverges to infinity, while the skyrmion energy converges to zero. This critical value coincides with the expected transition point from the uniform phase, which accommodates the skyrmion as an excited state, to the helical phase, which has negative energy.
We give the profile field at the skyrmion core, its outer field, and the intermediate field at the skyrmion domain wall.
Moreover, we derive an explicit formula for the leading asymptotic behavior of the energy as well as the leading term and first asymptotic correction for the value of the critical parameter.
The key leading to the results is a parity theorem that utilizes exact formulae for the asymptotic behavior of the solutions of the static Landau-Lifshitz equation centered at the skyrmion domain wall. 
The skyrmion energy is shown to be an odd function of the radius and the dimensionless parameter to be an even function.
\end{abstract}


\end{frontmatter}

\section{Introduction}
\label{sec:introduction}

Magnetic skyrmions are two-dimensional topological solitons
$\magn:\mathbb{R}^2 \cup \{\infty\} \to \mathbb{S}^2$ with
$\deg \magn= \pm 1$.
After their theoretical prediction \cite{BogdanovYablonskii_JETP1989,BogdanovHubert_JMMM1994} they have been observed in ferromagnets with the Dzyaloshinskii-Moriya (DM) interaction and techniques have been developed for individual skyrmions to be created and annihilated in a controlled manner \cite{RommingHanneken_Science2013}.
DM interaction arises from the loss of chiral symmetry induced by the underlying crystal structure or due to thin-film or multilayer geometries.
Chiral interaction terms and chiral skyrmions also arise in variational models for other condensed matter systems including spin-orbit coupled Bose-Einstein condensates (BEC) \cite{Aftalion_Mason2013, AftalionRodiac2020} or nematic liquid crystals \cite{Ackerman2014, Mermin_Wright}.

Our  model is based on a micromagnetic energy functional that includes exchange, DM and easy-axis anisotropy terms.
The system can be described by a single dimensionless DM parameter $\dmscaled$, defined in Eq.~\eqref{eq:epsilon}, given as the ratio of the DM parameter divided by (half) the domain wall energy. It is known (though not rigorously proven yet) that, there are only two phases minimizing the energy per unit area: the uniform phase and the helical phase for small and for large DM parameter $\dmscaled$ respectively \cite{BogdanovHubert_JMMM1994}.
The spiral state has negative energy and it is represented by one-dimensional (1D) modulations in the form of a distorted flat helix perpendicular to the helix propagation vector. 
This is a periodic solution of the 1D static Landau-Lifshitz equation. The transition from the spiral to the uniform state occurs at $\dmscaled=2/\pi$, and it is achieved as the period of the spiral goes to infinity for $\dmscaled\to2/\pi$. 

The isolated chiral skyrmion is an excited state in the parameter regime of non-negative energy where the uniform state is the absolute energy minimizer.
Most approaches are based on the assumption of axial symmetry so that $\magn$ is represented by its polar angle $\T=\T(r)$ depending on the radial coordinate $r>0$.
The existence of skyrmionic solutions as local minimizers of the micromagnetic energy has been rigorously proven for the case of an external field \cite{Melcher_PRSA2014, Li_Melcher_JFA2018}.
The argument has been extended to the case of uniaxial anisotropy including stray-field interaction \cite{BMS_arXiv, BMS_PRB}, and to director models of chiral liquid crystals \cite{Greco2019}.

Skyrmionic solutions of the static Landau-Lifshitz equation in the presence of a DM term can be found by numerical methods 
\cite{BogdanovHubert_JMMM1994,LeonovMonchestky_NJP2016}.
Numerical results provide the phase diagram for the existence of skyrmions and various features of the skyrmion profile.
The skyrmion profile determines to a large extent, and sometimes crucially, the skyrmion properties \cite{EverschorMasellReeveKlaeui_JAP2018}.
Its details are thus essential for the manipulation of individual skyrmions. Skyrmions exhibit different morphologies depending on the size of $\dmscaled$, see Figure~\ref{fig:smallLarge}.
For skyrmions of large radius, an ad-hoc ansatz based on explicit (1D) domain wall profiles \cite{Braun_PRB1994} has been suggested and is widely used to examine structural and dynamic properties, see, e.g., \cite{RommingKubetzka_PRL2015,Zhou_NCOMM_2015,BuettnerLemeshBeach_srep2018}.
In Ref.~\cite{KravchukShekaBrink_PRB2018} a 1D profile with the domain wall width as an additional parameter is used.
The profile enters in formulae for dynamical phenomena, for example, skyrmion translation and oscillation modes \cite{SchuetteGarst_PRB2014, KravchukShekaBrink_PRB2018} or antiferromagnetic skyrmion excitations \cite{KravchukGomonaySheka_PRB2019}, and it is crucial for quantitative calculations.
In recent years, sufficient resolution has been obtained for the observation of the features of the skyrmion profile in great detail \cite{RommingKubetzka_PRL2015,BoulleVogel_nnano2016,LeonovMonchestky_NJP2016,McGroutherLamb_NJP2016,KovacsBorkovski_PRL2017,ShibataTokura_PRL2017,MeyerPerini_NatComm2019}.
The availability of a  detailed analytical description of the skyrmion profile is thus important and it will open the way for a wider exploitation of individual skyrmions.

\begin{figure} 
    \centering
    (a)\includegraphics[width=6cm]{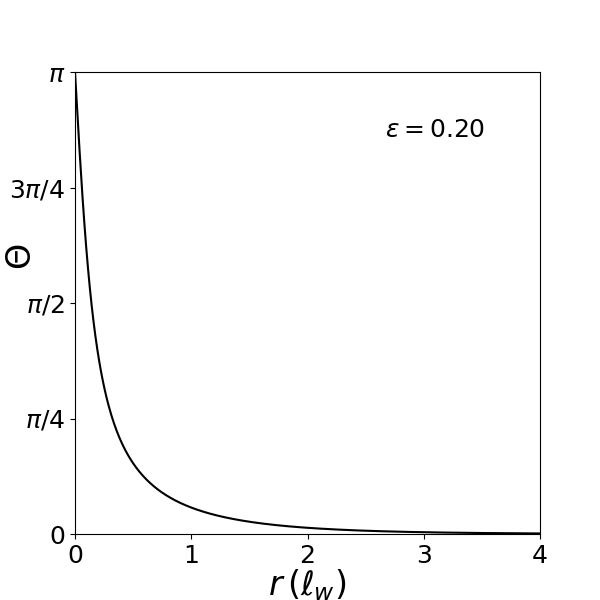}
    (b)\includegraphics[width=6cm]{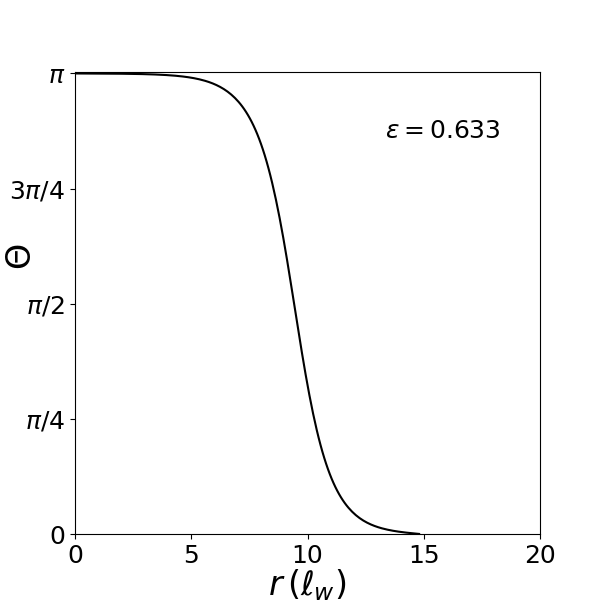}
    \caption{In contrast to the small skyrmion (a) being a localized perturbation of the Belavin-Polyakov soliton, the key feature of the large skyrmion (b) is a domain wall structure that separates the core from the far field.
The slope of the profile approaches $\pi$ exponentially in the limit of large radius.}
    \label{fig:smallLarge}
\end{figure}

For the case of skyrmions of small radius, analytic formulae for the profile of axisymmetric skyrmions have been derived  \cite{KomineasMelcherVenakides_arXiv2019}.
In this asymptotic regime where the dimensionless DM parameter $\dmscaled$ is small, magnetic skyrmions are well approximated on small scales by the classical Belavin-Polyakov soliton.
The results in Ref.~\cite{KomineasMelcherVenakides_arXiv2019} provide a quantitative description of this approximation in terms of asymptotic formulae for the skyrmion radius $R \sim \frac{\dmscaled}{|\ln \dmscaled|}$ and energy $E - 4\pi \sim \frac{\dmscaled^2}{\ln \dmscaled}$ for $\dmscaled \ll 1$.

In this paper, we derive formulae for the skyrmion profile in the case of large skyrmion radius by employing asymptotic methods that give analytic approximations of the skyrmion solutions for the time-independent Landau-Lifshitz equation. 
Our analysis predicts a breakdown of skyrmions solutions, via a diverging radius, when approaching the threshold value $\dmscaled=2/\pi$ from below, and thereby supports its role as critical constant.
We define the skyrmion radius $R$ via $\Theta(R)=\pi/2$ and we write the angle $\Theta$ as an asymptotic series
\begin{equation}  \label{eq:Theta_series1}
 \T = \T_0 + \frac{\T_1}{R}+\frac{\T_2}{R^2}+\frac{\T_3}{R^3} + \cdots.
\end{equation}
It is customary to define a variable $T=R-r$ which shifts the origin from the center to the radius of the skyrmion, and to consider $\T=\T(T)$. 

A key observation is the following parity property of the expansion \eqref{eq:Theta_series1}. We show that the functions $\T_n(T)$ are  odd functions of $T$ if $n$ is even; they are even functions, if $n$ is odd. 
The angle $\T_0$ coincides with the functional form of the 1D domain wall (Bloch wall), and the $\T_n,\,n=1,2,\ldots$ are asymptotic corrections.
We also show that the DM parameter $\dmscaled$ is expressed as the even asymptotic series 
\begin{equation}  \label{eq:eps-R0}
    \dmscaled = \dmscaled_0 + \frac{\dmscaled_2}{R^2} + \frac{\dmscaled_4}{R^4} + \cdots
\end{equation}
We obtain $\dmscaled_0=2/\pi$ (within the two-dimensional model) and this coincides with the value at which the transition from the ferromagnetic regime to the helical regime takes place.
The numerical values of the coefficients $\dmscaled_n,\,n=2,4,\cdots$ are calculated. Finally we show that the skyrmion energy $E$ is expressed as the odd asymptotic series 
\begin{equation}  \label{eq:energy-R}
E = \frac{E_1}{R} + \frac{E_3}{R^3}+\cdots
\end{equation}
The results of the present analysis for large radius, taken in combination with the results of Ref.~\cite{KomineasMelcherVenakides_arXiv2019} for small radius give a reasonably complete description of the skyrmion profile and energy depending on $\dmscaled$.

The paper is arranged as follows.
In Section~\ref{sec:equation}
we explain the mathematical model and present the equation for the skyrmion profile, while in subsection~\ref{sec:Core_outerRegion} we give the formulae for the skyrmion profile in the core and the outer region.
In Section~\ref{sec:higherOrder} we give a systematic method to obtain an asymptotic series for the skyrmion profile.
In Section ~\ref{sec:numerics} we apply the theory of the previous section and obtain numerical values for the asymptotic formulae.
In Section~\ref{sec:Pohozaev} we derive a Pohozaev identity and apply this to find explicit formulae for the $\dmscaled$ vs $R$ relation.
In Section~\ref{sec:energy} we give an asymptotic expansion for the skyrmion energy.
\ref{sec:core_outerRegin_app} contains the details of the calculations for the skyrmion profile in the core and in the outer region.
\ref{sec:parity} contains the proof of a theorem which establishes a fundamental parity property for the skyrmion profile.
\ref{sec:calculations} contains the details of the calculations for the $\dmscaled_2$ and for the energy expansion.

\section{The skyrmion equation}
\label{sec:equation}

We consider a two-dimensional ferromagnet on the $xy$-plane with exchange, Dzyaloshinskii-Moriya interaction, and anisotropy of the easy-axis type perpendicular to the plane.
The micromagnetic structure is described via the magnetization vector $\magn=\magn(x,y)$ with a fixed magnitude normalized to unity, $\bm{m}^2=1$.
The normalized form of the micromagnetic energy reads \cite{KomineasMelcherVenakides_arXiv2019}
\begin{equation} \label{eq:E0}
E_{\dmscaled}(\magn)=\int \left[ \hf \p_\mu\magn\cdot\p_\mu\magn + \hf(1-m_3^2) + \dmscaled\, e_{\rm DM} \right] \,\mathrm{d}x.
\end{equation}
A summation over repeated indices $\mu=1,2$ is assumed.
The last term in the parenthesis in Eq.~\eqref{eq:E0} models the DM interaction. Prototypical cases are the bulk DM interaction form $e_{\rm DM} = \bm{\hat{e}}_\mu\cdot (\p_\mu\magn\times\magn)$ and the interfacial DM interaction form $e_{\rm DM} = \emn \bm{\hat{e}}_\mu\cdot (\p_\nu\magn\times\magn)$, where $\emn$ is the totally antisymmetric two-dimensional tensor. 
Here $\ex,\ey,\ez$ are the unit vectors for the magnetization in the respective directions. Static magnetization configurations satisfy the static Landau-Lifshitz equation
\begin{equation} \label{eq:LL0}
\magn \times \left( \p_\mu\p_\mu\magn + m_3 \ez - 2\dmscaled\, \bf{h}_{\rm DM} \right) = 0.
\end{equation}
where the last term is the DM field with
$\bf{h}_{\rm DM} = \bm{\hat{e}}_\mu\times \p_\mu\magn$ in case of bulk interaction or $\bf{h}_{\rm DM} = \emn \, \bm{\hat{e}}_\mu\times \p_\nu\magn$ in case of interfacial DM. In Eq.~\eqref{eq:E0} and \eqref{eq:LL0}, lengths are measured in units of the domain wall width $\ldw = \sqrt{A/K}$, where $A$ is the exchange and $\Anisotropy$ the anisotropy constant.
The equation contains a single parameter
\begin{equation}  \label{eq:epsilon}
\dmscaled = \frac{\ldk}{\ldw} = \frac{\DM}{2\sqrt{A\Anisotropy}}
\end{equation}
defined via an additional length scale of this model $\ldk = \DM/(2\Anisotropy)$, where $\DM$ is the DM parameter (in Ref.~\cite{BogdanovHubert_JMMM1994}, a parameter which differs from $\dmscaled$ only by a constant factor has been introduced).
We will refer to $\dmscaled$ as the {\it dimensionless DM parameter}, but one should keep in mind that it can also be controlled by changing the anisotropy or the exchange parameter.
The lowest energy (ground) state is the spiral for $\dmscaled > 2/\pi$ and the ferromagnetic state for $\dmscaled < 2/\pi$ \cite{BogdanovHubert_JMMM1994}.

Let us consider the angles $(\Theta, \Phi)$ for the spherical parametrization of the magnetization vector, and the polar coordinates $(r,\phi)$ for the film plane.
We assume an axially symmetric skyrmion with $\Phi = \phi+\phi_0$ and $\Theta=\Theta(r)$.
For a bulk DM term the energy is minimized for $\phi_0=\pi/2$ (Bloch skyrmion) and for interfacial DM interaction we choose $\phi_0=0$ (N\'eel skyrmion).
A value $0<\phi_0<\pi/2$ should be chosen if the DM term is a combination of the bulk and interfacial terms.

The skyrmion profile arises as a local minimizer of the energy
\begin{equation}\label{eq:energy}
 E_{\dmscaled}(\magn)=2\pi\int_0^\infty \left[ \frac{1}{2}\left(\frac{\der\T}{\dr}\right)^2   
+\frac{1}{2}\left(1+\frac{1}{r^2}\right)\sin^2\T
+\dmscaled\left(\frac{\der\T}{\dr}
+\frac{1}{2r}\sin2\T\right)\right] r \dr
\end{equation}
of 
\[
\magn(r, \phi)=\left(\sin \T \cos(\phi+\phi_0),\sin \T \sin(\phi+\phi_0), \cos \T \right)
\]
whereby $\Theta=\Theta(r)$ satisfies the equation 
\begin{equation}  \label{eq:thetaODE}
\Theta''+\frac{\Theta'}{r}-\frac{\sin(2 \Theta)}{2 r^2} - \frac{\sin(2 \Theta)}{2}
+ 2\dmscaled \frac{\sin^2 \Theta}{r} = 0
\end{equation}
with boundary conditions $\Theta(0)=\pi$ and $\lim_{r \to \infty}\Theta(r)=0$.
The same equation applies to all types of skyrmions, e.g., Bloch and N\'eel skyrmions for the respective DM terms.

\subsection{The skyrmion core and the outer region}
\label{sec:Core_outerRegion}

We study skyrmions with large radius $R$, defined by the equation 
\[
\T(R)=\frac{\pi}{2}.
\]
The skyrmion profile exhibits three spatial regions.
The skyrmion core is the region where the value of $\T$ is close to $\pi$ (magnetization pointing close to the south pole).
The outer region (or far field) is where $\T$ is exponentially close to zero (magnetization pointing close to the north pole).
The skyrmion domain wall is the thin region that connects the core and the outer region.
Eq.~\eqref{eq:thetaODE} reduces to the modified Bessel equation both at the skyrmion core and in the far field and it is studied in \ref{sec:core_outerRegin_app}.
Using asymptotic analysis, we obtain the following results.
Close to the skyrmion center, the deviation of the skyrmion  profile from $\pi$ is linear with an exponentially small factor (see Eqs.~\eqref{eq:Theta_solution_region_A}, \eqref{eq:C_1_value}),
\begin{equation} \label{eq:veryCore}
\T \approx \pi - e^{-R}\sqrt{2\pi R}\, r, \qquad r \ll 1.
\end{equation}
As $r$ increases, the deviation attains exponential growth; this is held in check by the  small factor throughout the skyrmion core, up to the approach to the domain wall (see Eqs.~\eqref{eq:I1_asymptotic}, \eqref{eq:C_1_value}),
\begin{equation} \label{eq:core}
\T \approx \pi - 2\sqrt{\frac{R}{r}}\,e^{r-R},\qquad  1\ll r \ll R.
\end{equation}
The leading approximation of the skyrmion domain wall profile is independent of the radius when the radius is large, (see Sec.~\ref{sec:higherOrder}).
Past the domain wall, in the far field, the behavior is similar to the one of   skyrmions of small radius \cite{KomineasMelcherVenakides_arXiv2019}. We  have 
(see Eqs.~\eqref{eq:Theta_asymptotic_region_B}, \eqref{eq:C_2_value})
\begin{equation} \label{eq:farField}
\T \approx 2\sqrt{\frac{R}{r}}\,e^{-(r-R)},\qquad r \gg R.
\end{equation}
The core and the far field profiles are  matched with the respective sides of the domain wall profile to leading order.

\section{High order analysis of the skyrmion domain wall}
\label{sec:higherOrder}

We focus attention in the region of the skyrmion domain wall and develop an analysis valid to all orders in $R^{-1}$ for skyrmions of large radius.
The leading behavior of the solution as $R\to\infty$ is obtained by neglecting the terms of Eq.~\eqref{eq:thetaODE} with $r$ in the denominator.
The emerging equation 
\begin{equation}\label{eq:pendulum}
   \T''-\hf\sin( 2\Theta)=0
\end{equation}
characterizes the leading behavior of the domain wall of the skyrmion and has solution 
\begin{equation}\label{eq:pendulum_solution}
   \Theta_0=2\arctan\left(e^{-T}\right),
\end{equation}
where 
\begin{equation}  \label{eq:T_definition}
    T=r-R,   \ \ \ \ T\in(-R,\infty).
\end{equation}
The constants of integration follow from the requirements $\T_0=\frac{\pi}{2}$ when $T=0$ and $\T_0\to 0$ as $T\to \infty$. 
We calculate easily the following quantities that will be used below, \begin{equation}\label{eq:theta_0_related_quantities}
 \T_0'= -\sech T,\quad  \cos(2\T_0)=1-2\sech^2T,\quad
 \sin(2\T_0)=2\sech T\tanh T,\quad
 \sin^2\T_0=\sech^2T.
\end{equation}

Proceeding to a higher order analysis, we use the radius $R$ as the parameter of the problem.
We construct an asymptotic series for the profile $\Theta$ in negative powers of $R$ to all orders.
The profile $\T(T)$ is expanded to an asymptotic series for large $R$,
 \begin{equation}\label{eq:Theta_series}
     \T = \T_0 + \tilde{\T},\qquad 
     \tilde{\T} = \frac{\T_1}{R}+\frac{\T_2}{R^2}+\frac{\T_3}{R^3}+\cdots.
\end{equation}
$\T_0(T)$ is given by Eq.~\eqref{eq:pendulum_solution} and $\T_1, \T_2, \T_3,\cdots$ are also functions of $T$.
Necessarily $\dmscaled$ must be expressed in terms of the parameter $R$.
We choose the same form of asymptotic expansion as for $\Theta$, 
\begin{equation}  \label{eq:epsilon_expansion}
\dmscaled=\dmscaled_0+\frac{\dmscaled_1}{R}+\frac{\dmscaled_2}{R^2}+\frac{\dmscaled_3}{R^3}+\cdots.
\end{equation}
We introduce the expansion
\begin{equation}  \label{eq:radius_expansion}
    \frac{1}{r} = \frac{p}{T},\qquad p = \frac{T}{R}\left( 1-\frac{T}{R}+\frac{T^2}{R^2}+\cdots\right).
\end{equation}
The motivation for this notation is that $p$ is a power series of the ratio $T/R$.
We finally introduce the expansions of trigonometric functions isolating the leading order,
\begin{equation}  \label{eq:CS}
    \cos(2\tilde\T)=1+C(2\tilde\T), \qquad
    \sin(2\tilde\T)=2\tilde\T+S(2\tilde\T)
\end{equation}
where $C, S$ contain the higher order terms of the Taylor expansions about zero of the cosine and sine functions, respectively.
Inserting the series \eqref{eq:Theta_series} for $\T$ into Eq.~\eqref{eq:thetaODE}, applying the identities of trigonometric addition and using Eqs.~\eqref{eq:epsilon_expansion}, \eqref{eq:radius_expansion}, \eqref{eq:CS} obtains
\begin{equation}  \label{eq:Theta_n_ODE0}
   \tilde{\Theta}''-\cos(2\Theta_0)\tilde{\T}=\tilde{g}
\end{equation}
where the prime denotes differentiation with respect to $T$, and
\begin{equation}
    \tilde g=\frac{g_1}{R}+\frac{g_2}{R^2}+\frac{g_3}{R^3}+\cdots.
\end{equation}
The explicit form of $\tilde{g}$ is given in Eq.~\eqref{eq:thetaODE_expanded_form_3}.
The hierarchy of linear nonhomogeneous equations for the functions $\T_n$ is obtained directly from Eq.~\eqref{eq:Theta_n_ODE0},
\begin{equation}  \label{eq:Theta_n_ODE}
    \T_n''-(\cos2\T_0)\,\T_n=g_{n}, \qquad n=1, 2, 3,\cdots.
\end{equation}
The forcing term $g_{n}$ of the equation for $\T_n$ may depend only on the functions $\T_l$ with $l\le n-1$.
All equations have the same homogeneous part.
All equations are given the initial condition $\T_n(T=0)=0$.

The homogeneous equation corresponding to the hierarchy \eqref{eq:Theta_n_ODE} is
\begin{equation}  \label{eq:homogeneous}
  \T_H''-(1-2\sech^2T)\,\T_H=0.
\end{equation}
This equation describes the motion of a quantum mechanical particle in a potential well (see, e.g., Ref.~\cite{LandauLifshitz_QuantumMechanics}, page 73).
The potential equaling negative $\sech^2 T$ is one of the Bargmann reflectionless potentials, a class of potentials of the one-dimensional Schr\"odinger operator having bound states with negative energy and zero reflection coefficient for all positive energies \cite{NovikoVManakov}.
Eq.~\eqref{eq:homogeneous} has the explicit basis solutions
\begin{equation}\label{eq:H_1_H_2}
   H_1=\sech T, \qquad  H_2=\sinh T+T\sech T.
\end{equation}
Their Wronskian is given by
\begin{equation}
  \det \, \begin{pmatrix}
    H_1 & H_2\\ H_1' & H_2'
   \end{pmatrix}=2.
\end{equation}
Using the formula of the variation of constants, we obtain
\begin{equation}  \label{eq:variation_of_constants}
  \T_n = -\hf H_1(T)\int_0^T g_{n}(\tau) H_2(\tau)\,\dtau
   + \hf H_2(T)\int_{-\infty}^T g_{n}(\tau)H_1(\tau)\,\dtau,\quad n=1,2,3,\cdots.
\end{equation}
This solution satisfies the boundary condition $\Theta_n(0)=0$ and the solvability condition (boundary condition at infinity)
\begin{equation}  \label{eq:finitenessCondition}
    \int_{-\infty}^\infty g_{n}(\tau)H_1(\tau)\,\dtau = 0,\qquad n=1,2,3,\cdots.
\end{equation}
The condition is the consequence of the fact that  Eq.~\eqref{eq:Theta_n_ODE} for $\T_n$ has the form $\mathbb{L}\T_n=g_n$, where $\mathbb{L}$ is a selfadjoint differential operator.
Since $H_1(T)$ is in its nullspace, the inner product $(g_n,H_1)=(\mathbb{L}\T_n, H_1)=(\T_n,\mathbb{L}H_1$)=0.

The calculation of the $\T_n$ is recursive; in order to demonstrate the calculational pattern, we examine  the explicit form of the functions $g_1, g_2, g_3$
\begin{equation}  \label{eq:g1g2g3}
\begin{split}
    g_1 & = - (\T_0' + 2\dmscaled_0 \sin^2\T_0) \\
    g_2 & = T (\T_0' + 2\dmscaled_0 \sin^2\T_0) - 2 \dmscaled_1 \sin^2\T_0 + \sin 2\T_0 \left(\hf - 2\dmscaled_0 \T_1 - \T_1^2 \right) - \T_1' \\
    g_3 & = -T^2 (\T_0' + 2\dmscaled_0 \sin^2\T_0) + 2 (T\dmscaled_1 - \dmscaled_2)\sin^2\T_0 + T\T_1' - \T_2' \\
    & + \sin 2\T_0 \left[ -T + (T\dmscaled_0 - \dmscaled_1) 2\T_1 -2\dmscaled_0 \T_2 - 2 \T_1\T_2  \right]
    + \cos 2\T_0 \left( \T_1 - 2\dmscaled_0 \T_1^2 -\frac{2}{3} \T_1^3 \right).
\end{split}
\end{equation}
We make the following observations
\begin{enumerate}
    \item $\T_0$ is an odd function of $T$, thus, $g_1$ is even. Inserting $g_1$ into the solvability condition \eqref{eq:finitenessCondition}, produces the value of $\dmscaled_0$.
    \item The fact that $g_1$ in
     Eq.~\eqref{eq:variation_of_constants}
     is even implies that also $\T_1$ is even.
   \item All the terms of $g_2$ are odd with the exception of the term multiplied by $\dmscaled_1$, which is even. Inserting $g_2$ into the solvability condition \eqref{eq:finitenessCondition}, produces $\dmscaled_1=0$. Thus, $g_2$ is odd.
   \item The fact that $g_2$ in  Eq.~\eqref{eq:variation_of_constants} is odd implies that also $\T_2$ is odd.
   \item All terms of $g_3$ are even. Inserting $g_3$ into the solvability condition \eqref{eq:finitenessCondition} produces the value of $\dmscaled_2$.
\end{enumerate}
The cycle continues periodically according to the flow chart 
\[
\underbrace{\Theta_0}_{\hbox{odd}} \rightarrow \underbrace{g_1}_{\hbox{even}} \rightarrow 
\underbrace{\left\{ \begin{matrix} \dmscaled_0 \\ \Theta_1 \end{matrix} \right\}}_{\hbox{even}} \rightarrow
\underbrace{g_2}_{\hbox{odd}} \rightarrow 
\underbrace{\left\{ \begin{matrix} \dmscaled_1=0 \\ \Theta_2 \end{matrix} \right\}}_{\hbox{odd}} \rightarrow
\underbrace{g_3}_{\hbox{even}} \rightarrow 
\underbrace{\left\{ \begin{matrix} \dmscaled_2 \\ \Theta_3 \end{matrix} \right\}}_{\hbox{even}} \rightarrow
\underbrace{g_4}_{\hbox{odd}} \rightarrow 
\underbrace{\left\{ \begin{matrix} \dmscaled_3=0 \\ \Theta_4 \end{matrix} \right\}}_{\hbox{odd}} \rightarrow
\cdots
\]
with odd indexed $g_n$ and $\T_n$ being even and even indexed $g_n$ and $\T_n$ being odd.

The coefficient $\dmscaled_n$ makes its first appearance in the expression of $g_{n+1}$ multiplying the term $-2\sech^2 T$ for every $n$.
Using this, the solvability condition \eqref{eq:finitenessCondition} produces the values
\begin{equation}  \label{eq:finitenessCondition_explicit}
    \dmscaled_n = \frac{1}{\pi}\int_{-\infty}^\infty (\sech\tau)\,g_{n+1}(\tau)|_{\dmscaled_n=0}\,\dtau,\qquad n=0,1,2,3,\cdots.
\end{equation}
For every odd $n$ the integrand is odd giving $\dmscaled_n=0$, i.e., all odd indexed $\dmscaled_n$ vanish.
As a result, relation \eqref{eq:epsilon_expansion} of the dimensionless DM parameter with the skyrmion radius is simplified to
\begin{equation}  \label{eq:epsilon-R}
\dmscaled = \dmscaled_0 + \frac{\dmscaled_2}{R^2} + \frac{\dmscaled_4}{R^4} + \frac{\dmscaled_6}{R^6} + \cdots.
\end{equation}
The parity results stated here are proved in the following theorem.
\begin{theorem}  \label{thm:parity}
Let $\dmscaled_{2i-1}=0$ for $i=1,2,3,\cdots$.
Then, the following parity conditions hold.
\begin{enumerate}
   \item For all $n\ge1$, the functions $g_n=g_n(T)$ are even if $n$ is odd and they are odd if $n$ is even.
   \item The same is true for the functions $\T_n=\T_n(T)$, for $n\ge 0$.
\end{enumerate}
\end{theorem}
The theorem is instrumental for the following calculations.
Its proof, involving some subtlety, is relegated to \ref{sec:parity} in order to allow the flow of the calculation to be continued uninterrupted.

\section{Numerics}
\label{sec:numerics}

\begin{figure}
    \centering
    (a)\includegraphics[width=4.8cm]{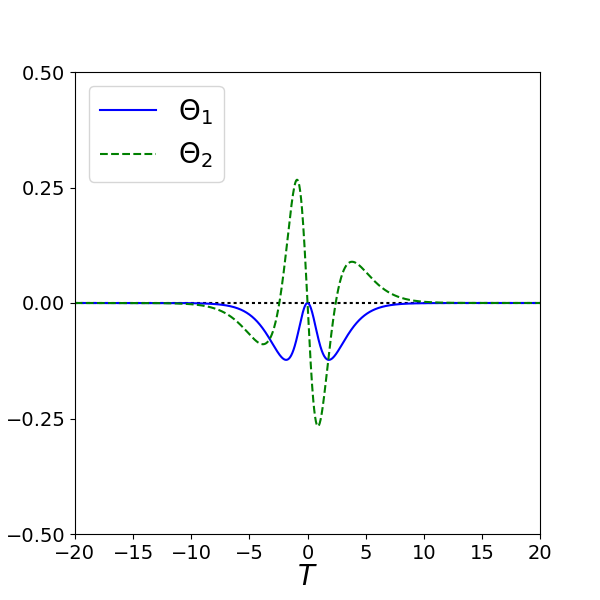} \hspace{5pt}
    (b)\includegraphics[width=4.8cm]{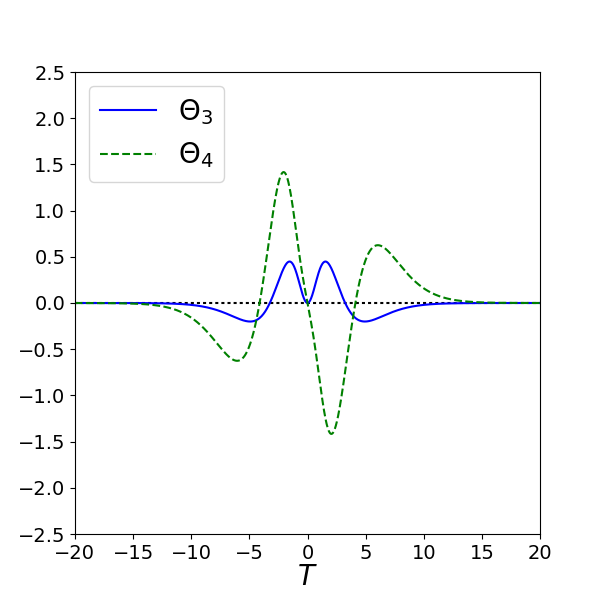} \hspace{5pt}
    (c)\includegraphics[width=4.8cm]{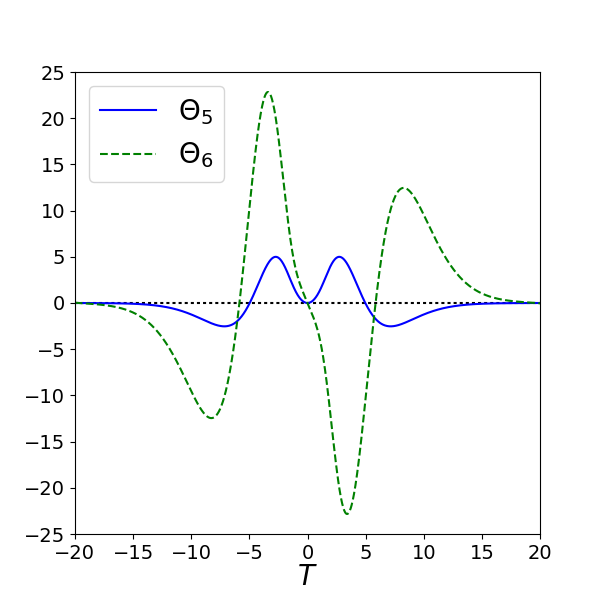}
    \caption{The functions (a) $\T_1,\, \T_2$, (b) $\T_3,\, \T_4$, and (c) $\T_5,\, \T_6$ calculated by numerical evaluation of the integrals in Eq.~\eqref{eq:variation_of_constants}.
    As seen by the change of scale of the vertical axis in the three entries, the values of the functions $\T_n$ increase fast with increasing index $n$.
    Note the even parity of the odd indexed functions and the odd parity of the even indexed ones with respect to the variable $T=r-R$.
    }
    \label{fig:thetas}
\end{figure}

We proceed to find $\T_1, \T_2, \T_3, \T_4, \T_5, \T_6$ by applying Eq.~\eqref{eq:variation_of_constants}.
The expressions for $g_n$ for higher $n$ are long and they have been derived using the mathematics software system SageMath \cite{sagemath}.
The derivatives $\Theta_n'$ needed in the expressions of $g_n$ are found by finite differences in the numerical calculation.
Fig.~\ref{fig:thetas} shows functions $\T_1$ through $\T_6$.
As expected from Theorem~\ref{thm:parity}, odd indexed $\T_n$'s are even functions of $T$ and even indexed ones are odd.
Functions $\T_n$ with higher index $n$ take higher values and they take significant values over larger intervals of $T$.
These features have consequences for the quality of the approximation, especially for small $R$, as we shall see in the following.

\begin{figure}
    \centering
    (a)\includegraphics[width=6cm]{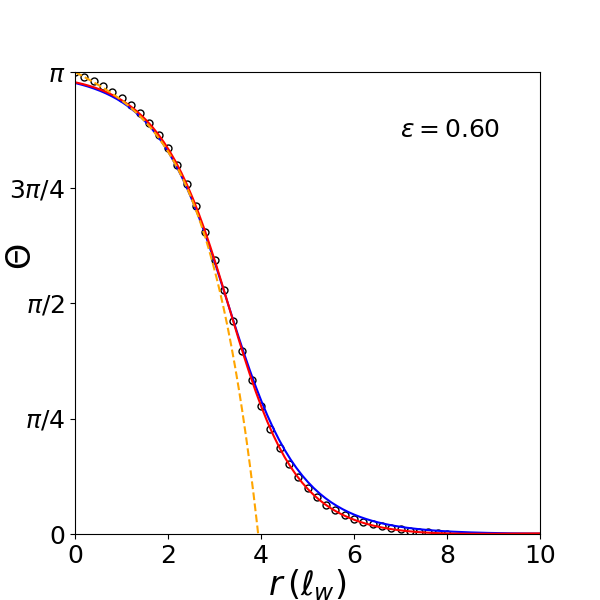} \hspace{10pt}
    (b)\includegraphics[width=6cm]{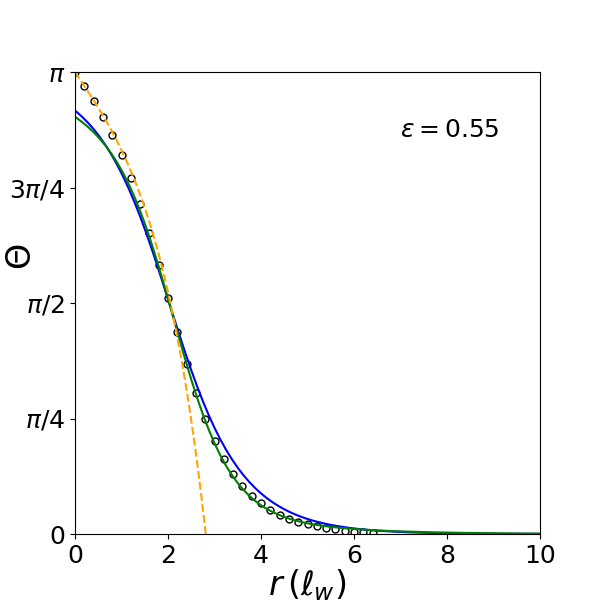}
    \caption{The profile of a skyrmion $\T(r)$ is shown by small circles for two values of the parameter $\dmscaled$, obtained numerically by solving the original Eq.~\eqref{eq:thetaODE} using a shooting method.
    The blue line shows $\T_0$, which is the one-dimensional domain wall profile.
    The orange line shows the solution \eqref{eq:Theta_solution_region_A} for the skyrmion profile at the core, obtained by linearizing the original equation about $\T=\pi$.
    (a) For $\dmscaled=0.60$ the skyrmion has a radius $R=3.29\,\ldw$, as obtained by the shooting method.
    The red line shows the series \eqref{eq:Theta_series} summed up to the term $\T_6$.
    (b) For $\dmscaled=0.55$ the skyrmion has a radius $R=2.02\,\ldw$.
    For this smaller value of $R$ the optimal point of truncation of the series occurs at the term $\T_2$.
    The green line shows the series \eqref{eq:Theta_series} summed up to and including the term $\T_2$.
    }
    \label{fig:profile}
\end{figure}

We have calculated by a shooting method the skyrmion profiles, which are solutions of Eq.~\eqref{eq:thetaODE}, for various values of the parameter $\dmscaled$. (We have actually solved an equation for the stereographic projection of the magnetization vector which is equivalent to Eq.~\eqref{eq:thetaODE}, as described in Ref.~\cite{KomineasMelcherVenakides_arXiv2019}.)
They are shown in Fig.~\ref{fig:profile} by small circles for two values of the parameter $\dmscaled$.
In Fig.~\ref{fig:profile}a, we have $\dmscaled=0.60$ that gives a skyrmion of radius $R=3.29\,\ldw$.
The profile at the skyrmion core is approximated very well by the solution given in Eq.~\eqref{eq:Theta_solution_region_A} with the coefficient given in Eq.~\eqref{eq:C_1_value}, and is shown as an orange dashed line in the figure.
The blue solid line shows the one-dimensional domain wall profile $\Theta_0$, shown in Eq.~\eqref{eq:pendulum_solution}, centred at the skyrmion radius position $r=R$.
The red line shows the series solution \eqref{eq:Theta_series} for $\Theta$ up to the term $O(1/R^6)$.
The approximation of the skyrmion profile is excellent for all $r$ except near the skyrmion center $r=0$.
In Fig.~\ref{fig:profile}b, we have $\dmscaled=0.55$ and a smaller skyrmion radius $R=2.02\,\ldw$.
The solution \eqref{eq:Theta_solution_region_A} still gives an excellent approximation at the skyrmion core.
The green line shows the series solution \eqref{eq:Theta_series} for $\Theta$ up to the term $O(1/R^2)$ and obtains a good approximation of the profile, especially around the skyrmion radius and for $r>R$.
The terms of the series \eqref{eq:Theta_series} of order higher than $O(1/R^2)$ cannot be used to improve the approximation, and they rather give larger deviations from the true profile if added to the series.
This phenomenon could have already been anticipated given the form of the $\T_n$'s shown in Fig.~\ref{fig:thetas} and the observation that an increasing index $n$ gives $\T_n$'s with rapidly increasing values.
When a term $\T_n/R^n$ is larger than the previous term in the series the series should be truncated omitting this term.

\begin{figure}
    \centering
    (a) \includegraphics[width=7cm]{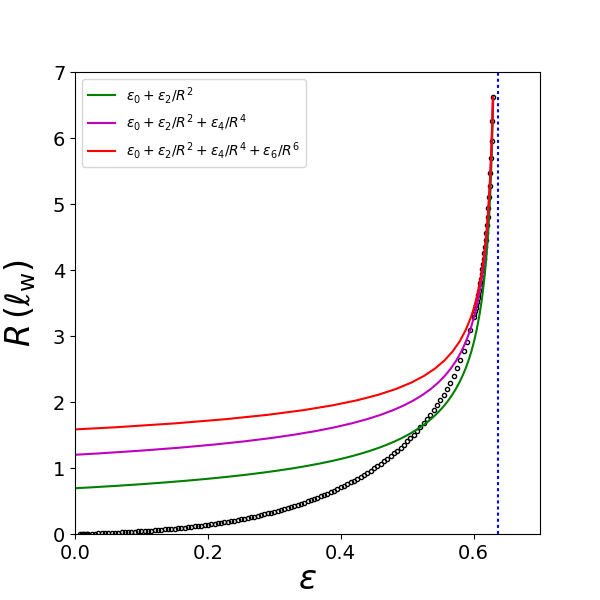} \hspace{10pt}
    (b) \includegraphics[width=7cm]{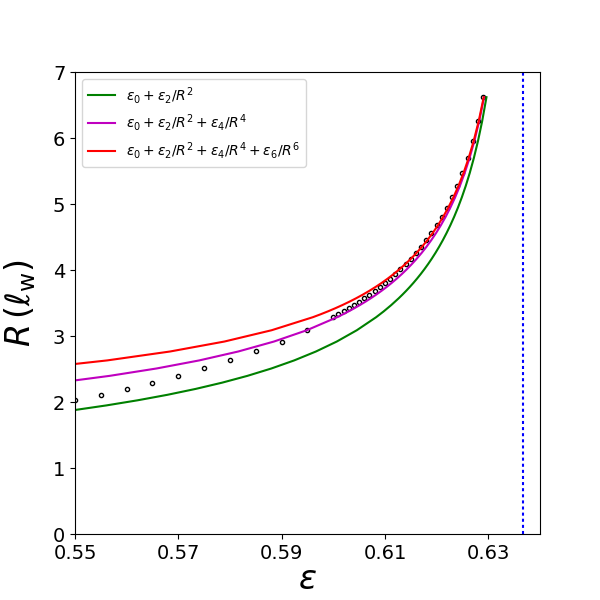}
    \caption{(a) The skyrmion radius $R$ found numerically by solving the original equation \eqref{eq:thetaODE} for various values of $\dmscaled$ is shown by open circles.
Relation \eqref{eq:eps-R_series} is shown by the colored lines to order $O(1/R^2),\, O(1/R^4),\,O(1/R^6)$ as indicated in the legend.
The dotted blue line is an asymptote and marks the critical value $\dmscaled=2/\pi$.
The successive approximations enhance the accuracy as $R$ increases.
(b) A blow-up of the graph for the region of large values of $R$.
}
    \label{fig:eps-R}
\end{figure}

We now calculate the numerical values of $\dmscaled_n$ from Eq.~\eqref{eq:finitenessCondition_explicit}.
We obtain
\begin{equation}  \label{eq:eps-R_series}
    \dmscaled \approx \frac{2}{\pi} - \frac{0.3057}{R^2} - \frac{0.8792}{R^4}  - \frac{5.901}{R^6},\qquad R \gg 1,
\end{equation}
where
\begin{equation}  \label{eq:epsilon0}
     \dmscaled_0=\frac{1}{\pi}\int_{-\infty}^\infty\sech^2\tau\, \dtau = \frac{2}{\pi}
\end{equation}
is obtained analytically while
\begin{equation}  \label{eq:epsilon2}
    \dmscaled_2 \approx -0.3057,\quad\dmscaled_4 \approx -0.8792,\quad
    \dmscaled_6 \approx -5.901
\end{equation}
are found by numerical integration.
Inverting Eq.~\eqref{eq:epsilon-R} we obtain the skyrmion radius versus the parameter $\dmscaled$,
\begin{equation}\label{eq:R(epsilon)_formula}
    R = \frac{|\dmscaled_2|^{1/2}}{\tilde{\dmscaled}^{1/2}} + O\left(\tilde{\dmscaled}^{1/2}\right),\qquad \tilde{\dmscaled} = \frac{2}{\pi}-\dmscaled.
\end{equation}

Fig.~\ref{fig:eps-R} shows by open circles the skyrmion radius extracted from the calculation of the skyrmion profiles by the shooting method.
These data are compared with formula \eqref{eq:eps-R_series} for the successive approximations up to and including order $O(1/R^2),\, O(1/R^4)$, and $O(1/R^6)$.
The approximation is excellent for large $R$ and it is improving as we add higher order terms, as seen in the blow-up in Fig.~\ref{fig:eps-R}b.
For smaller $R$, higher order approximations give larger deviations from the correct result, especially when the term $O(1/R^6)$ is included.
This is a consequence of the increasing error in the asymptotic series for small values of $R$ as was discussed in relation to Fig.~\ref{fig:profile}.

\section{A Pohozaev type identity and an explicit form of $\dmscaled_2$}
\label{sec:Pohozaev}

We multiply Eq.~\eqref{eq:thetaODE} by $2r\T'$ and we integrate over the $r$ axis.
After straightforward algebraic and trigonometric manipulations, we obtain
\begin{equation}
\int_0^\infty\left[\left(r\T'^2\right)'+\Theta'^2-\left(r+\frac{1}{r}\right)(\sin^2\T)'
+ 4\dmscaled (\sin^2 \Theta)\T'\right]\dr = 0.
\end{equation}
The first and last terms under the integral are exact derivatives; they integrate to zero and to $-2\pi\dmscaled$ respectively. 
Performing an integration by parts, we obtain the Pohozaev type identity that is satisfied by all skyrmions and is crucial for our calculation,
\begin{equation}\label{eq:skyrmion_law}
\int_0^\infty
\left[\Theta'^2
+\left(1-\frac{1}{r^2}\right)\sin^2\T\right]\dr 
= 2\pi\dmscaled.
\end{equation}

The following theorem, based on relation \eqref{eq:skyrmion_law}, provides an explicit formula for $\dmscaled_2$ and an alternative derivation of $\dmscaled_0$.

\begin{theorem}
\label{thm:explicit_epsilon2}
The DM parameter $\dmscaled$ satisfies the relation
\begin{equation}
    \dmscaled = \frac{2}{\pi} - \frac{1}{\pi R^2}
   \left( 1 + \hf
    \int_{-\infty}^\infty\T_1g_1 \dT \right) 
    +O\left(R^{-4}\right).
\end{equation}
\begin{proof}
We take the limit $R\to \infty$ in the integral law \eqref{eq:skyrmion_law}.
The integral of $\sin^2\T/r^2$ converges to zero; the integrand decays exponentially in the core and outer region of the skyrmion.
In the limit, we are left with 
\begin{equation}\label{eq:integral_for_epsilon0}
\int_{-\infty}^\infty\left( \Theta_0'^2+\sin^2\T_0
\right) \dT = 2\pi\dmscaled_0.
\end{equation}
The integral in Eq.~\eqref{eq:integral_for_epsilon0} is calculated using Eqs.~\eqref{eq:theta_0_related_quantities} and equals 4.
The critical value $\dmscaled_0=2/\pi$ follows directly.

For the determination of $\dmscaled_2$ we expand the Pohozaev identity \eqref{eq:skyrmion_law} in powers up to $O\left(R^{-2}\right)$.
We insert $\T=\T_0+\tilde\T$, and make $T$ the integration variable.
The manipulation, which involves integrations by parts resulting in significant cancellations, is relegated to \ref{sec:epsilon_calculations}.
We obtain
\begin{equation} \label{eq:epsilon2_integral}
    \dmscaled_2 = -\frac{1}{\pi} \left(1+\hf\int_{-\infty}^\infty\T_1 g_1 \dT \right).
\end{equation}
The value of the latter integral is found in \ref{sec:evaluation_T1g1}.
Inserting the value given in Eq.~\eqref{eq:T1g1_value} into Eq.~\eqref{eq:epsilon2_integral} we find
\begin{equation} \label{eq:epsilon2_value}
\dmscaled_2 = -\frac{0.9605}{\pi} \approx -0.3057
\end{equation}
It agrees with the result in Eq.~\eqref{eq:epsilon2} found by a different method in Sec.~\ref{sec:numerics}.
\end{proof}
\end{theorem}

In Ref.~\cite{RohartThiaville_PRB2013}, the value $\dmscaled_2=-1/\pi$ was found by an energy minimization argument that takes only $\Theta_0$ into account.

\section{Energy of skyrmion}
\label{sec:energy}

Let us denote by $\Eex, \Ea, \Edm$ the exchange, anisotropy and DM energy terms in the total energy \eqref{eq:energy}.
Using a standard scaling argument \cite{Derrick_JMP1964} one proves that any localized configuration that is a minimum of the energy in an infinitely extended two-dimensional system, such as a skyrmion, satisfies
\begin{equation}  \label{eq:virial2D}
    2\Ea + \Edm = 0.
\end{equation}
In a one dimensional system the same argument gives for minima of the energy, such as a domain wall,
\begin{equation}  \label{eq:virial1D}
    \Eex - \Ea = 0.
\end{equation}
In the limit $\dmscaled \to 2/\pi$, the latter relation is correct to leading order also for a skyrmion, as the leading approximation for the skyrmion profile, in this limit, is $\Theta_0$.
Thus, in the limit, a skyrmion satisfies both \eqref{eq:virial2D} and \eqref{eq:virial1D}, and these are combined to give
\begin{equation}
    \Eex = \Ea = -\frac{\Edm}{2}\qquad \hbox{when}\quad\dmscaled\to \frac{2}{\pi}.
\end{equation}
The total skyrmion energy is
\begin{equation}
    E = \Eex + \Ea + \Edm = 0\qquad \hbox{when}\quad\dmscaled\to\frac{2}{\pi}.
\end{equation}

We give a full asymptotic series for the energy in the following theorem.

\begin{theorem}
All three forms of skyrmion energy in Eq.~\eqref{eq:energy}, exchange, anisotropy and DM, have asymptotic expansions in which only odd powers of $R$ are present,
\begin{equation} \label{eq:EexEaEdm_expansion}
    \Eex = 2\pi R + \frac{{\Eex}_{,1}}{R} + \frac{{\Eex}_{,3}}{R^3} + \cdots,\;
    \Ea = 2\pi R + \frac{{\Ea}_{,1}}{R} + \frac{{\Ea}_{,3}}{R^3} + \cdots,\;
    \Edm = -4\pi R + \frac{{\Edm}_{,1}}{R} + \frac{{\Edm}_{,3}}{R^3} + \cdots,
\end{equation}
The three leading terms sum up to zero; the total energy is
\begin{equation}\label{eq:energy_expansion}
   E \sim \frac{E_1}{R} + \frac{E_3}{R^3}+ \frac{E_5}{R^5}+\cdots, \quad E_1=4\pi^2|\dmscaled_2|, \qquad\quad R\to\infty.
\end{equation}
The skyrmion energy tends to zero, as $\dmscaled$ increases, approaching its critical value, with a rate of convergence
\begin{equation} \label{eq:energy_vs_epsilon}
E \sim \left(4\pi^2|\dmscaled_2|^{1/2}\right)\, \tilde{\dmscaled}^{1/2} + O\left(\tilde{\dmscaled}^{3/2}\right),\quad \tilde{\dmscaled}= \frac{2}{\pi}-\dmscaled.
\end{equation}
\end{theorem}

\begin{proof}
We classify the functions of the form $f(T)R^n$, where the function $f$ can be odd or even and the power $n$ can be any integer, according to the following table. 

\begin{center}
 \begin{tabular}{||c c c ||} 
 \hline
 Class & $n$  & $f$   \\ [0.5ex] 
 \hline\hline
 A & odd & even  \\ 
 \hline
 B & even & odd  \\
 \hline
 C & odd & odd  \\
 \hline
 D & even & even  \\ [1ex] 
 \hline
\end{tabular}
\end{center}
It suffices to prove that every term of the expression under the energy integral \eqref{eq:energy} belongs to class $A+B+C$. The integration will eliminate the terms in $B$ and $C$; only odd powers of $R$ will participate in the expression of each of the three components of the energy and, hence, in the total energy.

In order to simplify the notation, we use the symbol of the class to also denote the class elements as well as sums of the class elements. Thus, 
\begin{equation} 
\T_0 +\frac{\T_2}{R^2}+\frac{\T_4}{R^4}+\cdots\equiv B, \ \ \ \ \
 \frac{\T_1}{R}+\frac{\T_3}{R^3}+ \frac{\T_5}{R^5}+\cdots\equiv A, \ \ \ \ \
T\equiv B, \ \ \ \ \ \ R\equiv A. 
\end{equation}
We notice that 
\begin{equation}
     A^2\equiv B^2\equiv  C^2\equiv D^2\equiv D,  \ \ \  AB\equiv C, \ \  \ AC=B, \ \ \ BC\equiv A, \ \ \ AD\equiv A, \ \ \ BD\equiv B, \ \ \  A'\equiv C,\ \  \ B'\equiv D.
\end{equation}

We also have 
\begin{equation}
    \sin A\equiv A, \ \ \ \  \sin B\equiv B,\ \ \ \  \cos A\equiv \cos B\equiv D.
\end{equation}
By applying trigonometric identities, we obtain 
\begin{equation}
 \sin(A+B)\equiv AD+BD\equiv A+B, \ \ \ \ \cos(A+B)\equiv DD+AB\equiv D+C.
\end{equation}
We perform the calculation term by term,
expressing $\sin^2\T$ as $\hf\left(1-\cos 2\T\right)$. 
\begin{enumerate} \itemsep4pt
    \item Term $r\left(\frac{\der\T}{\dr}\right)^2 \equiv(R+T)(A+B)^2
\equiv (A+B)(C^2+CD+D^2)\equiv (A+B)(D+C)\equiv A+B$.
\item Term $r\frac{\der\T}{\dr}\equiv (A+B)(C+D)\equiv A+B. $
\item Term $\sin2\T\equiv\sin(A+B)\equiv\sin(A+B)\equiv A+B.$
\item Term $r+\frac{1}{r}\equiv A+B$ \ \ (see Eq.~\eqref{eq:radius_expansion}).
\item Term $\left(r+\frac{1}{r}\right)\cos2\T\equiv (A+B)(D+C)\equiv A+B$.
\end{enumerate}
\smallskip
This shows that only odd powers of $R$ are present.

The highest order term in the energy is $O(R)$ with the contributions of the three energy terms (exchange, anisotropy and DM) given by the corresponding three terms in the integral
\begin{equation}\label{eq:energy0}
2\pi R \int_{-\infty}^\infty \left\{\frac{1}{2}\left(\frac{\der\T_0}{\dT}\right)^2
 +\frac{1}{2}\sin^2\T_0
+\frac{2}{\pi}\left(\frac{\der\T_0}{\dT}
\right)\right\} \dT.
\end{equation}
Inserting the relevant quantities from Eqs.~\eqref{eq:theta_0_related_quantities}, we obtain the $O(R)$ terms in Eq.~\eqref{eq:EexEaEdm_expansion} for the individual energies.
The $O(R)$ term in the total energy vanishes.
The term $O(R^{-1})$ is calculated in \ref{sec:energyCalculation}, giving the result of Eq.~\eqref{eq:energy_expansion}.
Inserting Eq.~\eqref{eq:R(epsilon)_formula} into Eq.~\eqref{eq:energy_expansion}, we obtain Eq.~\eqref{eq:energy_vs_epsilon}.
\end{proof}

\section*{Acknowledgement}
We are grateful to Stefan Bl\"ugel and to Alex Bogdanov for fruitful discussions.
CM gratefully acknowledges financial support by the DFG under the grant no. ME 2273/3-1, and SK a Mercator fellowship as part of the previous grant.
SV gratefully acknowledges financial support by the NSF through contract DMS-1211638.
SK acknowledges funding from the Hellenic Foundation for Research and Innovation (HFRI) and the General Secretariat for Research and Technology (GSRT), under grant agreement No 871.

\appendix

\section{The skyrmion core and the outer region}
\label{sec:core_outerRegin_app}

We derive the leading asymptotic behaviors of $\T$ in the skyrmion core ($-R<T<0, \ |T|\gg1$) and in the outer region ($T>0, \ |T|\gg1$) and match them with the expressions of $\T$ obtained for the domain wall. Matching occurs on the overlap layers  $1\ll|T|\ll R$, with $T<0$ on the left and $T>0$ on the right. The leading asymptotic on the overlap layers obtained from the domain wall's  Eq.~\eqref{eq:pendulum_solution} is 
\begin{equation}  \label{eq:Theta_leadingAsymptotic}
    \T\sim
    \begin{cases} \pi-2e^{-|T|},  \ \ \ T<0 \\ \\
    2e^{-|T|}, \qquad T>0.\end{cases}.
\end{equation}

\subsection*{The skyrmion core (Region A)}

In the  spatial region from $r=0$ and up to the domain wall, we have 
\begin{equation}
    \t(r)=\pi-\T(r)\ll1, \ \ \ \ \ \ \sin\t\sim\t.
\end{equation}
The DM term in Eq.~\eqref{eq:thetaODE} is clearly subdominant ($\theta^2\ll \theta$) and is neglected.
In terms of $\t$, Eq.~\eqref{eq:thetaODE} becomes the modified Bessel equation
\begin{equation}
 \label{eq:thetaODE_before_domain_wall}
r^2\theta''+r\theta'-(r^2+1)\theta=0, \qquad \theta(0)=0.   
\end{equation}
We will use its series solution $I_1(r)$ which equals zero at $r=0$ (this is the modified Bessel function of the first kind, see \cite{AbramowitzStegun}, par. 9.6.10),
\begin{equation}\label{eq:Theta_solution_region_A}
    \t = C_1 I_1(r), \ \ \ \ I_1(r)= \frac{r}{2}\sum_{n=0}^\infty
    \frac{\left(\frac{1}{4}r^2\right)^{n}}{n!(n+1)!},
\end{equation}
where $C_1$ is a constant.
The large $r$ behavior of the function $I_1(r)$ is 
(see \cite{AbramowitzStegun}, par. 9.7.1)
\begin{equation}\label{eq:I1_asymptotic}
 I_1(r)\sim  \frac{e^r}{\sqrt{2\pi r}}\left(1-\frac{3}{8r} +\cdots\right).
\end{equation}
The constant $C_1$ is now evaluated by identifying the  leading asymptotic deviation from $\pi$ of the angle $\T$ on the overlapping layer of the region A and the domain wall,
\begin{equation}\label{eq:Theta_core_and_wall}
\T(r)\sim
\begin{cases} \pi-C_1 \frac{e^r}{\sqrt{2\pi r}}
\sim\pi-C_1\frac{e^R}{\sqrt{2\pi R}}e^{-|T|}, \ \ \ T<0,\, |T| \gg 1\; \text{(Region A, Modified Bessel)}, \\ \\
\pi-2e^{-|T|}, \qquad\qquad\qquad\qquad\quad |T|\ll R\; \text{(Domain wall)}.
\end{cases}
\end{equation}
We obtain
\begin{equation}\label{eq:C_1_value}
    C_1 = e^{-R}\sqrt{8\pi R}. 
\end{equation}
The slope of the skyrmion profile at its center $r=0$ is
\begin{equation}  \label{Theta_slope_at_origin}
    \frac{d\T}{dr}(r=0)= -e^{-R}\sqrt{2\pi R}. 
\end{equation}

\subsection*{The skyrmion outer region (region B)}

As in region A, the DM term may be neglected. The angle $\T(r)$ then satisfies the modified Bessel equation \eqref{eq:thetaODE_before_domain_wall}.
The appropriate solution is $K_1(r)$ (modified Bessel function of the second kind) which decays as $r\to\infty$. Since $r\gg1$ in this region, we only need the asymptotic behavior of $K_1(r)$ (see \cite{AbramowitzStegun}, par. 9.7.2)
\begin{equation}\label{eq:Theta_asymptotic_region_B}
\T=C_2 K_1(r),\qquad  K_1(r)\sim \sqrt{\frac{\pi}{2r}}e^{-r}\left(1-\frac{3}{8r}
 +\cdots\right),
\end{equation}
where $C_2$ is a constant.

The constant $C_2$ is  evaluated by identifying the  leading behavior of the angle $\T$ on the overlapping layer of the region B and the domain wall,
\begin{equation}
\T(r)\sim
\begin{cases} 
C_2\sqrt{\frac{\pi}{2r}}e^{-r} \sim C_2\sqrt{\frac{\pi}{2R}}e^{-R}e^{-|T|}\,, \ \ \ T>0,\, T\gg 1\; \text{(Region $B$, Modified Bessel)}, \\ \\
2e^{-|T|}\,, \qquad\qquad\qquad\qquad\qquad T\ll R\; \text{(Domain wall)}.
\end{cases}
\end{equation}

Matching the leading order terms of the formulae for $\T$ in the overlapping region gives 
\begin{equation}\label{eq:C_2_value}
C_2=\sqrt{\frac{8R}{\pi}}e^R.
\end{equation}

\section{Parity theorem}
\label{sec:parity}

Proceeding to the calculation at higher orders of $1/R$ we determine the formula for $\tilde{g}$ utilizing Eqs.~\eqref{eq:radius_expansion}, \eqref{eq:CS}.
We obtain
\begin{equation}  \label{eq:thetaODE_expanded_form_3}
\begin{split}
\tilde{g} & =
-p\frac{\Theta_0'}{T}
-p\dmscaled\frac{1-\cos(2\Theta_0)}{T}
+p^2\frac{\sin(2\Theta_0) }{2T^2}
-p\frac{\tilde\Theta'}{T}
+\left(-p\dmscaled\frac{\sin(2\T_0)}{T} + p^2\frac{\cos(2\Theta_0)}{2T^2} \right) 2\tilde{\T}  \\
 & +\left(\frac{1}{2}\sin(2\Theta_0) + p\dmscaled\frac{\cos(2\Theta_0)}{T} + p^2\frac{\sin(2\Theta_0)}{2T^2} \right) C(2\tilde\T)
 +\left(\frac{1}{2}\cos(2\Theta_0) - p\dmscaled\frac{\sin(2\T_0)}{T} + p^2\frac{\cos(2\Theta_0)}{2T^2} \right) S(2\tilde\T).
\end{split}
\end{equation}
The following theorem is instrumental for all calculations based on Eq.~\eqref{eq:Theta_n_ODE}, and in particular for those presented in Sec.~\ref{sec:higherOrder}.
The hypotheses of the theorem turn out to be necessary conditions for the existence of bounded solutions $\T_n$.

\begin{theorem*}
Let $\dmscaled_{2i-1}=0$ for $i=1,2,3,\cdots$.
Then, the following parity conditions hold.
\begin{enumerate}
   \item For all $n\ge1$, the functions $g_n=g_n(T)$ are even if $n$ is odd and they are odd if $n$ is even.
   \item The same is true for the functions $\T_n=\T_n(T)$.
\end{enumerate}
\end{theorem*}
{\it Proof.}
When $g_n$ is odd, it follows directly from Eq.~\eqref{eq:variation_of_constants} that $\Theta_n$ is also odd.
When $g_n$ is even, the first two terms on the right side of Eq.~\eqref{eq:variation_of_constants} are even; the last term is odd and exponentially increasing therefore $c_{n,2}$ must be set to zero.
As a result $\Theta_n$ is even.
Therefore, it suffices to prove the theorem for the $g_n$.
We already know that the theorem is true for $n=1$.

Clearly, only terms $\T_j$ with $j<n$ appear in the expression for $g_n$, so we can truncate $\tilde\T$ accordingly.
We make the inductive assumption that the functions $g_1,g_2,g_3,\ldots g_{n-1}$ and hence the functions $\T_1,\T_2,\T_3,\ldots \T_{n-1}$ alternate in parity, with $g_1$ and $\T_1$ being even functions of $T$.
We prove that the theorem is then true for  $g_1,g_2,g_3,\ldots g_{n}$ and hence for $\T_1,\T_2,\T_3,\ldots \T_{n}$.
We recall that $\T_0$ is an odd function of $T$.
We define for convenience
\[
\delta = \frac{1}{R}
\]
and we observe that $p$ is a power series of $(\delta T)$.

We show that $g_n$ satisfies the parity condition term by term.
The condition is true for the first term, namely, $p\T_0'/T$.
Indeed, $p$ must be represented by $(\delta T)^n$ for the term to be of order $\delta^n$.
The term becomes $T^{n-1}\T_0'$, which satisfies the parity condition since $\T_0'$ is even.
In a similar way, the parity law is satisfied for the remaining two terms in which no $\T_j$ appears.
The verification for the terms in which only one $\T_j$ appears is equally straightforward. For example, if in the term before the first parenthesis, $p$ is represented by $(\delta T)^m$, the term is $T^{m-1}\T_{n-m}'\delta^n$.
The parity requirement is clearly satisfied for $m=1$, since taking the derivative changes the parity.
Increasing the value of $m$ by $k$ introduces $k$ factors $T$ and also shifts the index of $\T$ backwards
by $k$ positions.
According to our inductive assumption, the parity of the term is preserved.

More work is required to show that the parity condition holds for the terms that have the factor $C(\tilde\T)$ or $S(\tilde\T)$. 
According to the inductive assumption, the truncated
\begin{equation}
   \tilde\T=\delta\T_1+\delta^2\T_2+\delta^3\T_3+\ldots +\delta^{n-1}\T_{n-1}
\end{equation}
has even functions of $T$ multiplying the odd powers of $\delta$ and odd functions of $T$ multiplying the even powers of $\delta$. 
The general term of the expansions of $C(\tilde\T)$ and $S(\tilde\T)$ is represented by
\begin{equation}
  \tilde\T^{\textbf{k}}=(2\delta)^{\bf k\cdot q}\T_1^{k_1}\T_2^{k_2}\cdots\T_{n-1}^{k_{n-1}},
\end{equation}
where
${\textbf{k}}=(k_1,k_2,k_3,\cdots k_{n-1})$ with $k_i\in\{0,1,2,3,\cdots\}$, and
where ${\textbf{q}}=(1,2,3,\cdots, n-1)$.
The term $\tilde\T^{\textbf{k}}$ is either even or odd, since the factors $\T_j$ are even or odd.

\underbar{Claim}.
\begin{enumerate}
{\it  \item[(a)] All the terms in the expansion of $C(\tilde\T)$ are even at even powers of $\delta$ and odd at odd powers of $\delta$.
\item[(b)] All the terms in the expansion of $S(\tilde\T)$ are odd at even powers of $\delta$ and even at odd powers of $\delta$.}
\end{enumerate}

In order to prove the claim, we utilize the notion of the parity of a number or a function.
The parity equals zero in the case of evenness and it equals unity in the case of oddness of the number or function.
We calculate the following three parities.
\begin{enumerate}
   \item[(i)] The parity of the exponent $n$ in the order $\delta^n$ of a term $\tilde\T^{\bf k}$.
   \item[(ii)] The parity of the product $\T_1^{k_1}\T_2^{k_2}\cdots\T_{n-1}^{k_{n-1}}$ in  $\tilde\T^{\bf k}$.
  \item[(iii)] The parity of the number of the factors  $\T_j$ in $\tilde\T^{\bf k}$, counting multiplicities. 
  \end{enumerate}
We obtain the following. 
\begin{enumerate}
    \item[(i)] The exponent of $\delta$ equals $n={\bf k\cdot q}$. For the calculation of the parity of $n$, we set $k_jq_j=0$ if $q_j$ is even (eliminates all the odd factors $\T_j$) or if $k_j$ is even (eliminates all factors $\T_j$ with even multiplicity). 
    We are thus, keeping only the even factors with odd multiplicity.
    The parity is thus
    \begin{equation}
        N_{even}\equiv \text{\# of factors $\T_j$ that are even functions with odd multiplicity mod(2)}.
    \end{equation} 
    \item[(ii)] The second parity, which we denote by $N_{odd}$, equals 
   \begin{equation}
        N_{odd}\equiv \text{\# of odd factors with odd  multiplicity mod(2)}.
   \end{equation}  
   \item[(iii)] The parity of the number of the factors $\T_j$ in $\tilde\T^{\bf k}$, counting multiplicities is given by the sum $k_1+k_2+\cdots +k_{n-1}$. It is an even number for the terms of  $C(\tilde\T)$ and an odd number for the terms of  $S(\tilde\T)$. 
   The third parity, which we denote by $N_{total}$, equals 
   \begin{equation}
       N_{total}\equiv \text{\# of all factors of odd multiplicity mod(2).}
   \end{equation}  
   \end{enumerate}
Clearly, $N_{even}+N_{odd}\equiv N_{total} \mod(2)$. 
Hence, $N_{total}\equiv{0}$ in the case of $C(\tilde\T)$ and $ N_{total}\equiv{1}$ for $S(\tilde\T)$.
Parities 1 and 2 agree with each other in the terms of $C(\tilde\T)$ and differ from each other in the terms of $S(\tilde\T)$.
This proves the claim.

Now that the parity of the terms of $C$ and $S$ are understood, the correctness of the theorem for the terms involving these is verified similarly to the previous terms.

\section{Detailed calculations for $\dmscaled_2$ and for the energy expansion}
\label{sec:calculations}

\subsection{A formula for $\dmscaled_2$}
\label{sec:epsilon_calculations}

For the determination of $\dmscaled_2$ we expand the identity \eqref{eq:skyrmion_law} in powers up to $O\left(R^{-2}\right)$.
Inserting $\T=\T_0+\tilde\T$, and passing to $T$ as the integration variable, we obtain
\begin{equation}\label{eq:law_2}
\int_{-\infty}^\infty
\left[\left(\Theta_0'+\tilde\T'\right)^2 + \left(1-\frac{1}{R^2}\right)\sin^2\left(\Theta_0 + \tilde\T\right)\right] \dT 
= 2\pi\left(\dmscaled_0 + \frac{\dmscaled_2}{R^2}\right) + O\left(R^{-4}\right).
\end{equation}
We calculate
\begin{equation}\label{eq:Theta_sum_prime_squared}
 \left(\Theta_0'+\tilde\T'\right)^2
 =\Theta_0'^2+\frac{1}{R} 2\T_0'\T_1' + \frac{1}{R^2} \left( 2\T_0'\T_2' + \T_1'^2 \right) + O\left(R^{-3}\right)
\end{equation}
and
\begin{equation}\label{eq:sine_sum_sum_squared}
\begin{aligned}
 \sin^2\left(\Theta_0 + \tilde\T\right)
 &=\left(\sin\tilde\T\cos\T_0
 +\cos\tilde\T\sin\T_0\right)^2\\
& = \sin^2\T_0 + \frac{1}{R}\, \T_1\sin 2\T_0
+\frac{1}{R^2} \left( \T_1^2 \cos2\T_0 + \T_2\sin 2\T_0 \right) + O\left(R^{-3}\right).
\end{aligned}
\end{equation}
We insert these results into Eq.~\eqref{eq:law_2}.
The terms $O(1)$ cancel due to Eq.~\eqref{eq:integral_for_epsilon0}.
The $O\left(R^{-1}\right)$ terms are odd and vanish upon integration.
The terms $O(R^{-3})$ vanish for the same reason.
The $O\left(R^{-2}\right)$ terms give 
\begin{equation}
 \int_{-\infty}^\infty \left(
 2\T_0'\T_2'+\T_2\sin2\T_0+\T_1'^2+\T_1^2\cos2\T_0-\sin^2\T_0\right)\dT
 =2\pi\dmscaled_2.
\end{equation}
We perform an integration by parts in the first and third terms,
\begin{equation}
 \int_{-\infty}^\infty \left[
 \T_2(\sin2\T_0-2\T_0'') - \T_1(\T_1''-\T_1\cos2\T_0)-\sin^2\T_0\right] \dT
 =2\pi\dmscaled_2.
\end{equation}
The first parenthesis vanishes due to Eq.~\eqref{eq:pendulum} and the second one is equal to $g_1$ due to Eq.~\eqref{eq:Theta_n_ODE}. 
We obtain 
\begin{equation}
- \int_{-\infty}^\infty \left( \T_1g_1+\sin^2\T_0\right)\dT
 = 2\pi\dmscaled_2.
\end{equation}
We integrate the second term in the left side by using Eq.~\eqref{eq:theta_0_related_quantities} and obtain
\begin{equation} \label{eq:epsilon2_integral_app}
    \dmscaled_2 = -\frac{1}{\pi} \left(1+\hf\int_{-\infty}^\infty\T_1 g_1 \dT \right)
\end{equation}
proving a result in Theorem~\ref{thm:explicit_epsilon2}.
The value of the latter integral is found in \ref{sec:evaluation_T1g1}.

\subsection{A formula for the leading order energy}
\label{sec:energyCalculation}

We address the exchange, anisotropy and DM terms under the energy integral \eqref{eq:energy} separately.
The arrows below indicate that only terms $O\left(R^{-1}\right)$ are taken.
They also allow replacement of a term by an equal quantity, omitting terms that integrate to zero, and  operations corresponding to integration by parts, for example $\T_1'^2$ being replaced by $-\T_1''\T_1$.

\medskip
\begin{itemize}
\item Exchange:
\begin{equation}
\begin{aligned}
     \hf(R+T)\T'^2+\hf\frac{\sin^2\T}{R+T} : \ \  &\mapsto\hf\T_1'^2+\T_0'\T_2'+T\T_0'\T_1' +\hf\sin^2\T_0 \\
&\mapsto -\hf\T_1''\T_1-\T_0''\T_2-T\T_0''\T_1 -\T_0'\T_1+\hf\sin^2\T_0.
\end{aligned}
\end{equation}

\item
Anisotropy:
\begin{equation}
 \hf(R+T)\sin^2\T: \ \ \mapsto \hf\T_1^2\cos2\T_0
 +\hf\T_2\sin2\T_0+\hf T\T_1\sin2\T_0.
\end{equation}

\item
DM:
\begin{equation}
\begin{aligned}
  (R+T)\left(\dmscaled_0\T'+\frac{\dmscaled_2}{R^2}\T'\right)+\hf\dmscaled_0 \sin2\T: \ \ &\mapsto \dmscaled_0\T_2' + \dmscaled_0 T\T_1' + \dmscaled_2\T_0' + \dmscaled_0\,\T_1\cos2\T_0 \\
  &\mapsto - \dmscaled_0 \T_1 + \dmscaled_2\T_0' + \dmscaled_0\,\T_1\cos2\T_0 \\
  &\mapsto \dmscaled_2\T_0' - 2\dmscaled_0\,\T_1\sin^2\T_0.
 \end{aligned}
\end{equation}
\end{itemize}

The second and third terms in the anisotropy cancel one by one the corresponding exchange terms in view of Eq.~\eqref{eq:pendulum}.
The first term in the anisotropy combines with the first term in the exchange to give $-\hf\T_1 g_1$ in view of \eqref{eq:Theta_n_ODE}.
We are left with
\[
\begin{aligned}
& -\hf\T_1 g_1 - (\T_0'+2\dmscaled_0 \sin^2\T_0)\T_1 + \dmscaled_2\T_0' + \hf\sin^2\T_0 \\
\mapsto & \hf\T_1 g_1 + \dmscaled_2\T_0'+ \hf\sin^2\T_0.
\end{aligned}
\]
We take the integral
\[
\int_{-\infty}^\infty \left( \hf\T_1 g_1  + \dmscaled_2\T_0'+ \hf\sin^2\T_0\right) \dT
= \hf \int_{-\infty}^\infty \T_1 g_1\,\dT - \dmscaled_2\pi + 1
\]
and use \eqref{eq:epsilon2_integral} to obtain
\begin{equation} \label{eq:E1}
  E_1 = 4\pi \left( 1 + \hf \int_{-\infty}^\infty \T_1 g_1\, \dT \right) = 4\pi^2|\dmscaled_2|.
\end{equation}

\subsection{Evaluation of the integral in the expressions for $\dmscaled_2$ and $E_1$}
\label{sec:evaluation_T1g1}

We will evaluate the integral appearing in Eqs.~\eqref{eq:epsilon2_integral}, \eqref{eq:epsilon2_integral_app} for $\dmscaled_2$ and in Eq.~\eqref{eq:E1} for $E_1$.
For $n=1$, the solvability condition \eqref{eq:finitenessCondition} is equivalent to 
\begin{equation} \label{eq:finitenessCondition_even}
    \int_0^\infty g_1(\tau) H_1(\tau)\,\dtau = 0,
\end{equation}
as a result of the evenness of the integrand.
It follows that the lower limit $-\infty$ in Eq.~\eqref{eq:variation_of_constants} may be replaced with a zero lower limit;
the integral in Eq.~\eqref{eq:epsilon2_integral} is then written as
\begin{equation} \label{eq:T1g1}
    \int_{-\infty}^\infty \T_1g_1 \dT = -\hf \int_{-\infty}^\infty \dT\,g_1(T) H_1(T) \int_0^T \dtau\,g_1(\tau) H_2(\tau)
   + \hf \int_{-\infty}^\infty \dT\,g_1(T) H_2(T) \int_0^T \dtau\, g_1(\tau)H_1(\tau).
\end{equation}
The two terms on the right in Eq.~\eqref{eq:T1g1} are shown to be equal if we apply integration by parts  and use the solvability condition \eqref{eq:finitenessCondition_even}.
We have
\begin{equation} \label{eq:T1g1_short}
    \int_{-\infty}^\infty \T_1 g_1 \dT =
    \int_{-\infty}^\infty \dT\,g_1(T) H_2(T) \int_0^T \dtau\, g_1(\tau)H_1(\tau).
\end{equation}
The integral on the right can be calculated explicitly,
\begin{equation}
    \int_0^T g_1(\tau) H_1(\tau)\,\dtau = \tanh T - \frac{2}{\pi} \tanh T \sech T - \frac{4}{\pi} \arctan(e^T) + 1.
\end{equation}
This is inserted into Eq.~\eqref{eq:T1g1_short} and the integral is evaluated numerically to find
\begin{equation} \label{eq:T1g1_value}
\int_{-\infty}^\infty \T_1 g_1 \dT = -0.0790.
\end{equation}

\bigskip

\bibliographystyle{unsrt}
\bibliography{references}

\begin{thebibliography}{10}

\bibitem{BogdanovYablonskii_JETP1989}
A.~N. Bogdanov and D.~A. Yablonskii.
\newblock Thermodynamically stable ``vortices'' in magnetically ordered
  crystals. {T}he mixed state of magnets.
\newblock {\em Sov. Phys. JETP}, 68:101--103, 1989.

\bibitem{BogdanovHubert_JMMM1994}
A.~N. Bogdanov and A.~Hubert.
\newblock Thermodynamically stable magnetic vortex states in magnetic crystals.
\newblock {\em J. Magn. Magn. Mater.}, 138:255, 1994.

\bibitem{RommingHanneken_Science2013}
N.~Romming, C.~Hanneken, M.~Menzel, J.~E. Bickel, B.~Wolter, K.~von Bergmann,
  A.~Kubetzka, and R.~Wiesendanger.
\newblock Writing and deleting single magnetic skyrmions.
\newblock {\em Science}, 341(6146):636--639, 2013.

\bibitem{Aftalion_Mason2013}
Amandine Aftalion and Peter Mason.
\newblock Phase diagrams and thomas-fermi estimates for spin-orbit-coupled
  bose-einstein condensates under rotation.
\newblock {\em Phys. Rev. A}, 88:023610, Aug 2013.

\bibitem{AftalionRodiac2020}
Amandine Aftalion and R{\'e}my Rodiac.
\newblock One dimensional phase transition problem modeling striped spin orbit
  coupled bose-einstein condensates.
\newblock {\em Journal of Differential Equations}, 269:38--81, 2020.

\bibitem{Ackerman2014}
Paul~J. Ackerman, Rahul~P. Trivedi, Bohdan Senyuk, Jao van~de Lagemaat, and
  Ivan~I. Smalyukh.
\newblock Two-dimensional skyrmions and other solitonic structures in
  confinement-frustrated chiral nematics.
\newblock {\em Phys. Rev. E}, 90:012505, Jul 2014.

\bibitem{Mermin_Wright}
David~C. Wright and N.~David Mermin.
\newblock Crystalline liquids: the blue phases.
\newblock {\em Rev. Mod. Phys.}, 61:385--432, Apr 1989.

\bibitem{Melcher_PRSA2014}
C.~Melcher.
\newblock Chiral skyrmions in the plane.
\newblock {\em Proc. R. Soc. A}, 470:20140394, October 2014.

\bibitem{Li_Melcher_JFA2018}
X.~Li and C.~Melcher.
\newblock Stability of axisymmetric chiral skyrmions.
\newblock {\em J. Funct. Anal.}, 275(10):2817--2844, 2018.

\bibitem{BMS_arXiv}
Anne Bernand-Mantel, Cyrill~B Muratov, and Thilo~M Simon.
\newblock A quantitative description of skyrmions in ultrathin ferromagnetic
  films and rigidity of degree $\pm 1$ harmonic maps from $\mathbb{R}^2$ to
  $\mathbb{S}^2$.
\newblock {\em arXiv preprint arXiv:1912.09854}, 2019.

\bibitem{BMS_PRB}
Anne Bernand-Mantel, Cyrill~B. Muratov, and Thilo~M. Simon.
\newblock Unraveling the role of dipolar versus dzyaloshinskii-moriya
  interactions in stabilizing compact magnetic skyrmions.
\newblock {\em Phys. Rev. B}, 101:045416, Jan 2020.

\bibitem{Greco2019}
Carlo Greco.
\newblock On the existence of skyrmions in planar liquid crystals.
\newblock {\em Topological Methods in Nonlinear Analysis}, 2019.

\bibitem{LeonovMonchestky_NJP2016}
A~O Leonov, T~L Monchesky, N~Romming, A~Kubetzka, A~N Bogdanov, and
  R~Wiesendanger.
\newblock The properties of isolated chiral skyrmions in thin magnetic films.
\newblock {\em New Journal of Physics}, 18:065003, 2016.

\bibitem{EverschorMasellReeveKlaeui_JAP2018}
Karin Everschor-Sitte, J.~Masell, R.~M. Reeve, and M.~Kläui.
\newblock Perspective: Magnetic skyrmions—overview of recent progress in an
  active research field.
\newblock {\em J. Appl. Phys.}, 124:240901, 2018.

\bibitem{Braun_PRB1994}
H.-B. Braun.
\newblock Fluctuations and instabilities of ferromagnetic domain-wall pairs in
  an external magnetic field.
\newblock {\em Physical Review B}, 50(22):16485, 1994.

\bibitem{RommingKubetzka_PRL2015}
N.~Romming, A.~Kubetzka, C.~Hanneken, K.~von Bergmann, and R.~Wiesendanger.
\newblock Field-dependent size and shape of single magnetic skyrmions.
\newblock {\em Phys. Rev. Lett.}, 114:177203, May 2015.

\bibitem{Zhou_NCOMM_2015}
Y.~Zhou, E.~Iacocca, A.~A. Awad, R.~K. Dumas, F.~C. Zhang, H.-B. Braun, and
  J.~{\AA}kerman.
\newblock Dynamically stabilized magnetic skyrmions.
\newblock {\em Nature communications}, 6:8193, 2015.

\bibitem{BuettnerLemeshBeach_srep2018}
F.~B\"uttner, I.~Lemesh, and S.~D.~G. Beach.
\newblock Theory of isolated magnetic skyrmions: From fundamentals to room
  temperature applications.
\newblock {\em Sci. Rep.}, 8:4464, 2018.

\bibitem{KravchukShekaBrink_PRB2018}
Volodymyr~P. Kravchuk, Denis~D. Sheka, Ulrich~K. R\"o\ss{}ler, Jeroen van~den
  Brink, and Yuri Gaididei.
\newblock Spin eigenmodes of magnetic skyrmions and the problem of the
  effective skyrmion mass.
\newblock {\em Phys. Rev. B}, 97:064403, Feb 2018.

\bibitem{SchuetteGarst_PRB2014}
C.~Sch\"utte and M.~Garst.
\newblock Magnon-skyrmion scattering in chiral magnets.
\newblock {\em Phys. Rev. B}, 90:094423, Sep 2014.

\bibitem{KravchukGomonaySheka_PRB2019}
Volodymyr~P. Kravchuk, Olena Gomonay, Denis~D. Sheka, Davi~R. Rodrigues, Karin
  Everschor-Sitte, Jairo Sinova, Jeroen van~den Brink, and Yuri Gaididei.
\newblock Spin eigenexcitations of an antiferromagnetic skyrmion.
\newblock {\em Phys. Rev. B}, 99:184429, May 2019.

\bibitem{BoulleVogel_nnano2016}
Olivier Boulle, Jan Vogel, Hongxin Yang, Stefania Pizzini, Dayane {de Souza
  Chaves}, Andrea Locatelli, Tevfik~Onur Menteş, Alessandro Sala, Liliana~D.
  Buda-Prejbeanu, Olivier Klein, Mohamed Belmeguenai, Yves Roussign{\'{e}},
  Andrey Stashkevich, Salim~Mourad Ch{\'{e}}rif, Lucia Aballe, Michael
  Foerster, Mairbek Chshiev, St{\'{e}}phane Auffret, Ioan~Mihai Miron, and
  Gilles Gaudin.
\newblock Room-temperature chiral magnetic skyrmions in ultrathin magnetic
  nanostructures.
\newblock {\em Nat. Nano.}, 11(5):449--454, may 2016.

\bibitem{McGroutherLamb_NJP2016}
D~McGrouther, R~J Lamb, M~Krajnak, S~McFadzean, S~McVitie, R~L Stamps, A~O
  Leonov, A~N Bogdanov, and Y~Togawa.
\newblock Internal structure of hexagonal skyrmion lattices in cubic
  helimagnets.
\newblock {\em New Journal of Physics}, 18(9):095004, sep 2016.

\bibitem{KovacsBorkovski_PRL2017}
A.~Kov{\'a}cs, J.~Caron, A.~S. Savchenko, N.~S. Kiselev, K.~Shibata, Zi-An Li,
  N.~Kanazawa, Y.~Tokura, S.~Bl{\"u}gel, and R.~E. Dunin-Borkowski.
\newblock Mapping the magnetization fine structure of a lattice of {B}loch-type
  skyrmions in an {FeGe} thin film.
\newblock {\em Applied Physics Letters}, 111(19):192410, 2017.

\bibitem{ShibataTokura_PRL2017}
K.~Shibata, A.~Kov{\'a}cs, N.~S. Kiselev, N.~Kanazawa, R.~E. Dunin-Borkowski,
  and Y.~Tokura.
\newblock Temperature and magnetic field dependence of the internal and lattice
  structures of skyrmions by off-axis electron holography.
\newblock {\em Physical Review Letters}, 118(8):087202, 2017.

\bibitem{MeyerPerini_NatComm2019}
Sebastian Meyer, Marco Perini, Stephan {von Malottki}, Andr\'e Kubetzka, Roland
  Wiesendanger, Kirsten {von Bergmann}, and Stefan Heinze.
\newblock Isolated zero field sub-10 nm skyrmions in ultrathin {C}o films.
\newblock {\em Nat. Comm.}, 10:3823, Aug 2019.

\bibitem{KomineasMelcherVenakides_arXiv2019}
Stavros Komineas, Christof Melcher, and Stephanos Venakides.
\newblock The profile of chiral skyrmions of small radius.
\newblock {\em arXiv}, page 1904.01408, 2019.

\bibitem{LandauLifshitz_QuantumMechanics}
L.~D. Landau and E.~M. Lifshitz.
\newblock {\em Quantum Mechanics}.
\newblock Pergamon Press, Oxford, second edition, 1965.

\bibitem{NovikoVManakov}
S.~Novikov, S.~V. Manakov, L.~P. Pitaevskii, and V.~E. Zhakharov.
\newblock {\em Theory of solitons}.
\newblock Plenum Publishing Corporation, New York, 1984.

\bibitem{sagemath}
{The Sage Developers}.
\newblock {\em {S}ageMath, the {S}age {M}athematics {S}oftware {S}ystem
  ({V}ersion 8.7)}, 2019.
\newblock {\tt https://www.sagemath.org}.

\bibitem{RohartThiaville_PRB2013}
S.~Rohart and A.~Thiaville.
\newblock Skyrmion confinement in ultrathin film nanostructures in the presence
  of {D}zyaloshinskii-{M}oriya interaction.
\newblock {\em Phys. Rev. B}, 88:184422, Nov 2013.

\bibitem{Derrick_JMP1964}
G.~H. Derrick.
\newblock Comments on nonlinear wave equations as models for elementary
  particles.
\newblock {\em J. Math. Phys.}, 5:1252, 1964.

\bibitem{AbramowitzStegun}
Milton Abramowitz and Irene~A. Stegun.
\newblock {\em Handbook of Mathematical Functions with Formulas, Graphs, and
  Mathematical Tables}.
\newblock Dover, New York, ninth dover printing, tenth gpo printing edition,
  1964.

\end{thebibliography}

\end{document}